\documentclass[12pt]{article}
\usepackage[utf8]{inputenc}
\usepackage{tikz}
\usepackage{amsfonts} 
\usepackage{setspace}
\usepackage[plain]{fullpage}
\usepackage{amsmath,amsfonts,amssymb,amsthm}
\usepackage[authoryear, longnamesfirst]{natbib}
\usepackage{mathrsfs}  
\usepackage{calligra}
\usepackage{hyperref}
\usepackage{url} 
\usepackage{adjustbox}
\usepackage{subcaption}
\usepackage{easyReview}
\usepackage{adjustbox}
\usepackage{booktabs}
\usepackage[shortlabels]{enumitem}
\theoremstyle{definition}
\newtheorem{defn}{\protect\definitionname}
\providecommand{\definitionname}{Definition}

\theoremstyle{remark}
\newtheorem{rem}{\protect\remarkname}
\providecommand{\remarkname}{Property}

\newtheorem{assumption}{Assumption}

\providecommand{\lemmaname}{Lemma}
\newtheorem{prop}{\protect\propname}
\providecommand{\propname}{Proposition}

\providecommand{\theoremname}{Theorem}
\newcommand{\htheta}{\hat{\theta}}
\newcommand{\prob}{\mathbb{P}}
\newcommand{\E}{\mathbb{E}}

\newcommand{\tildeu}{\tilde{u}}
\newcommand{\tildem}{\tilde{\mathfrak{m}}}
\newcommand{\m}{\mathfrak{m}}
\newcommand{\bigcurlybrackets}[1]{\big\{{#1}\big\}}

\newcommand{\bigsquarebrackets}[1]{\big[#1\big]}

\newcommand{\bigbraces}[1]{\big(#1\big)}
\newcommand{\Bigbraces}[1]{\Big(#1\Big)}
\newcommand{\Biggbraces}[1]{\Bigg(#1\Bigg)}

\title{How Robust are Robustness Checks?\footnote{I would like to thank Jörg Stoye, Francesca Molinari and José Luis Montiel Olea for their guidance. I would also like to thank Levon Barseghyan, Edoardo Bollati, Senan Hogan-Hennessy, Yiqi Liu, Michael Lovenheim, Douglas Miller, Yiwei Sun, Vítor Possebom and seminar participants of the Econometric Society World Congress 2025 and Sociedade Brasileira de Econometria 2025 for valuable feedback. All errors are my own.} }
\author{Brenda Prallon\thanks{Cornell University. Department of Economics. \href{mailto:bq45@cornell.edu}{\textit{bq45@cornell.edu}}}}
\date{\today}

\begin{document}
	
	\maketitle
	
	\begin{abstract}
		Robustness checks are routine in empirical work, but there is no standard statistical procedure to formally measure what one can learn from them. I propose a ``robustness radius" measure to quantify the amount by which the robustness checks estimands differ from the main specification estimand. I do so by framing robustness checks as explicitly biased regressions, clarifying what exactly the estimands are when comparing multiple regressions with slightly different samples, and applying a test from the moment inequalities literature. The robustness radius is easily interpretable and adapts to sampling uncertainty and correlation across regressions. An application shows that, although assessing overall robustness is context-specific, the robustness radius guides those judgments and improves transparency.
	\end{abstract}
	
	\section{Introduction}
\label{sec:introduction}

Robustness checks are commonly carried out in empirical research. Even surveying a single journal, \textit{AEJ: Applied}, for a single year, $2023$, I find that of the $56$ papers published, $52$ explicitly report results of robustness, and the other $4$ carry out exercises that could arguably be considered robustness checks. 
Examples include dropping potentially ``problematic" observations, changing sample weights, or adding extra controls to address ex-post concerns about endogeneity, among others. Researchers' degrees of freedom in these choices may lead them to very different results even if using the same data (\cite{pnasuncertainty}). Robustness checks provide some insurance by slightly perturbing these choices and showing results do not change by ``a lot". 

The issue is ``a lot" is usually not formally quantified; often researchers conclude results are ``qualitatively similar". A naive but simple measure would take the maximum distance between robustness checks estimates and the main estimate; but it would not take into account any uncertainty in the estimators nor their correlation structure, and may end up being too conservative. The main contribution of this paper is to provide a measure that is intuitive, uses all information in the covariance matrix, and is based on a formal test of hypothesis with statistical guarantees. The goal is to replace multiple appendix tables and informal statements about robustness by an actual statistic, easily interpretable in each context. 

The proposed measure is the \textit{robustness radius}. For the main specification estimand $\theta_0$ and robustness checks estimands $\theta_j$, $j = 1, \ldots, m$, all scalars, consider the test 
\begin{align}
\label{eq:null_def}
	H_0 : \max_{j \in \{1, \ldots, m\}}|\theta_0 - \theta_j|  \leq b,
\end{align}
where $b \geq 0$ is also a scalar. The robustness radius is
 defined in words as the smallest distance $b$ for which $H_0$, the null of ``all robustness checks falling within this distance from the main specification," cannot be rejected. It is formally defined in Section \ref{sec:framework}. The robustness radius has the following features: 
 \begin{itemize}
 	\item It is formally defined and straightforward to interpret. 
 	\item It does not require $\theta_j  = \theta_0 \,\,\, \forall j$.
 	\item It adapts to sampling uncertainty of the estimators. 
 	\item It adapts to correlation structure of the estimators.
 \end{itemize}
 
 A robustness test has been proposed by \cite{lu2014robustness}, with a similar null from Equation \ref{eq:null_def} except taking $b=0$. This is a strong requirement: all robustness checks estimands must be identifying the exact same as the main specification estimand\footnote{For important cases where this assumption should hold, see \cite{lu2014robustness}.}. The test is equivalent to a Hausman test, and the authors propose using a more efficient estimator upon failure to reject. However, much more realistically, robustness checks may not identify the true parameter of interest, but still be best linear predictors that are close to the truth: in other words, it is often the case where $b>0$. To perform such a test (and hence obtain the robustness radius) I use the procedures described on \cite{coxandshi}. While other tests from the moment inequalities literature could potentially have been used\footnote{See \cite{MOLINARI2020355}, \cite{tamerPartialid}, \cite{canay2017practical} and \cite{CANAY2023105558} for extensive reviews.}, \cite{coxandshi} is fast to compute\footnote{Similarly to building confidence intervals by test inversion, to compute the robustness radius, the test must be repeatedly evaluated with different values of $b$.}, involve no tuning parameters, and includes, as a special case, the test of \cite{lu2014robustness}.

The prevalence of robustness checks can be linked to concerns about model misspecification. For recent reviews, see \cite{armstrong2024adaptingmisspecification}, \cite{andrews2025purposeestimatordoesmisspecification} and references therein. This paper does not consider the main specification $\theta_0$ to be misspecified, or partially identified; it considers robustness checks $\theta_j$'s to be misspecified, in the sense of typically assuming $\theta_j \neq \theta_0$, which seems to align with current empirical practice; it then asks the question of how to quantify the distance between them. If the researcher is worried about misspecification (or uncertainty) in $\theta_0$ and wishes to use robustness checks to expand the model in the sense of \cite{armstrong2024adaptingmisspecification}, $\theta_0$ should be considered partially identified and methods for union bounds such as in \cite{xinyuebei2024} can be used to compute appropriate confidence intervals (intuitively, the null in Equation \ref{eq:null_def} would have a flipped sign). 

Robustness checks are, often, a handcraft art; not only do they depend on the question at hand, but often also on specific datasets (for example, in Section \ref{sec:application}, I discuss an application that runs analyses across different weighting schemes due to a lack of documentation of survey weights)\footnote{Certain type of robustness checks can be systematic: see \cite{broderick2023automaticfinitesamplerobustnessmetric} for the case of checking sensitivity to few observations.}. The innate particularities of robustness checks indicate their limitations in measuring bias beyond specifications that can be observed. In fact, there is a large body of literature that has criticized the practice of informal robustness analyses as proxies for assessing omitted variable bias (\cite{imbens2003sensitivity}, \cite{altonjicatholic}, \cite{oster2019unobservable}, \cite{cinelli2020making}, \cite{chernozhukov2022long}, \cite{diegert2023assessingomittedvariablebias}, for a non-exhaustive list of examples). Although bias from unobservables is a very relevant topic to applied work, this paper focuses instead on the role of robustness checks in measuring bias from ``observables": minor choices, from data cleaning to economic/statistical modeling, that are deemed not to be major drivers of the main result. Nonetheless, Section \ref{subsec:sensitivity} shows how the robustness radius can be used to formally link bias from observables to bias from unobservables, and clarifies under which assumptions robustness checks can be used to guide sensitivity parameter choices.

 The remainder of the paper is divided as follows: Section \ref{sec:framework} introduces the framework and formal definition for robustness radius; Section \ref{sec:simulations} discusses the behavior of the robustness radius across different scenarios (well vs. not well separated estimands) and covariance structures; Section \ref{sec:application} shows one application and discusses the interpretation of the measure; 
 Section \ref{subsec:sensitivity} discusses the connection between the robustness radius and other formal sensitivity measures; and Section \ref{sec:conclusion} concludes.

	\section{Set-up}
\label{sec:framework}

Consider the problem of estimating the effect of variable $D \in \mathbb{R}$ on outcome $y \in \mathbb{R}$. 
Let $X \in \mathbb{R}^p$ be a set of possible covariates. Fix a finite set of $m + 1$ specifications: $\mathcal{M} = \{M_j\}_{j=0}^m$, where $M_j \subseteq \{1, \ldots p\}$ denotes an index set of coordinates of $X$, such that the researcher runs regressions of the form 
\begin{align}
	\label{eq:regs}
	y_j = \alpha_j + \theta_j D_j + \sum_{k \in M_j} \beta_{jk} X_k + e_j.
\end{align}
Here, outcome $y$ and variable of interest $D$ are indexed by $j$ to accommodate robustness checks that possibly use different proxies for $y$ or $D$. $\theta_j$ simply has the interpretation of best linear predictor if causal assumptions are not satisfied, with $e_j$ being the orthogonal residual.

Estimators $\{\htheta_j\}_{j = 0}^m$ are the ordinary least squares (OLS) estimators for $\theta_j$ in each of the $m + 1$ specifications. Without loss of generality, let $j = 0$ denote the ``main specification", that is, the specification for which the researcher argues about their identification strategy and is willing to attribute causal meaning to $\theta_0$. The remaining regressions $j = 1, \ldots m$ represent robustness checks. Their coefficients will only have causal meaning if $\theta_j = \theta_0, \,\, j \neq 0$. Notice there are no a priori reasons for this to be true, so without further assumptions, $\theta_j, \,\, j \neq 0$ is simply the population projection of $y_j$ on $D, X_{M_j}$, with $X_{M_j}$ being the $M_j$ coordinates of $X$. However, we believe that robustness checks should be informative, such that $\theta_j$ is close to $\theta_0$, $j \neq 0$. The question is: \textit{how} close?

To answer this, denote $\theta$ as the stacked $\theta_j$'s and consider, for a fixed $\theta \in \mathbb{R}^{m + 1}$, the following level $\alpha$ test $\phi_n(b, \alpha)$ (the dependency on $\theta$ is omitted for simplicity):
\begin{align}
	\label{eq:null_def2}
	H_0 : \max_{j \in \{1, \ldots, m\}}|\theta_0 - \theta_j|  \leq b,
\end{align}
where $b \geq 0$ is a scalar; $H_1$ is the complement of $H_0$.

In this form, $b$ can be interpreted as the maximum distance between the robustness checks estimands and the main estimand. 
I propose the definition of an intuitive measure for this distance, the \textit{robustness radius}:
\begin{defn}[Robustness radius]
	\label{def:rob_rad}
	The robustness radius, denoted $b_{RR}(\alpha)$, is given by:
	\begin{align}
		b_{RR}(\alpha) := \min\left\{b \geq 0: \phi_n(b, \alpha) = 0\right\}.
	\end{align}
	In words, for fixed $\theta$, $b_{RR}$ is the smallest value $b$ for which $H_0$ is not rejected at level $\alpha$. Equivalently, $b_{RR}(\alpha)$ is the lower limit of a one-sided $(1-\alpha)100\%$ confidence interval for $ \max_{j \in \{1, \ldots, m\}}|\theta_0 - \theta_j|$.
\end{defn}

If $b$ is very large, such that none of the inequalities are violated, any reasonable test would not reject $H_0$, but $b$ might not be very informative. If $b = 0$, this means we are assuming every specification is identifying the exact same estimand, which may be too strong of a requirement. On the other hand, $b_{RR}$ is defined as an adaptive quantity, similar in spirit to a p-value. It does not ask the researcher to simply take the largest distance between robustness checks and the main specification, nor does it require all robustness checks to be indistinguishable from the main estimate. It takes into account whether certain robustness checks have a larger variance, and how they correlate with the main estimate (and each other). 

I show that the null hypothesis in Equation \ref{eq:null_def2} can be tested using \cite{coxandshi}. Before diving into the specifics of the test, I highlight a subtlety that must be accounted for. Robustness checks often are based on using different covariates across checks. Due to missing data patterns or data combination limitations, this may lead to using different sub-samples of the data. It may also be the case that the main regression does not use the full sample and has deliberately dropped a few observations that are seen as ``problematic", but one robustness check is to include the problematic samples to ensure the main result is not extremely sensitive to them. In order to properly accommodate these cases, we need notation for missing data. 

Let $d_j$ be an indicator variable for the sub-population where specification $j$ is run, that is, $d_{ji} = 1$ if all variables $(D_j, X_{M_j})$ have been observed for observation $i$, and $0$ otherwise. To ease notation, I will refer to $(D_j, X_{M_j})$ as simply $X_j$ henceforth. Let $n$ be the total number of observations across all regressions (not double-counting). Let $n_j$ be the total number of observations that are used in regression $j$, and $I_j = \bigcurlybrackets{i: i \in \{1, \ldots, n\} \text{ and } d_{ji} = 1}$ be set of indexes of observations that regression $j$ uses. Then, OLS estimators are actually estimating sub-sample quantities, that is, 
\begin{align*}
	\htheta_j &= \underbrace{\Biggbraces{\frac{1}{n_j}\sum_{i \in I_j} X_{ji} X_{j i}'}^{-1}}_{\xrightarrow{p} \E[X_j X_j' | d_j = 1]^{-1}} \underbrace{\Biggbraces{\frac{1}{n_j}\sum_{i \in I_j} X_{j i} y_{ji}}}_{\xrightarrow{p} \E[X_j y_j | d_j = 1]}
	\xrightarrow{p} \theta_{j | d_j = 1},
\end{align*}
where $\theta_{j | d_j = 1}$ is the parameter $\theta_j$ for the sub-population from which the covariates used in regression $j$ are sampled from. 

To elucidate interpretation, consider $\theta_{0|d_0 = 1}$ to be the causal estimand as argued by the researcher. This would be, in the example of deliberately dropping ``problematic" observations, the effect of $D$ on $y_0$ in the sub-population of non-problematic observations. The distinction matters because taking the leap to claim $\theta_{0|d_0 = 1} = \theta_0$ would require a ``data missing completely at random" assumption which is, by construction in this example, not true. In the example where adding an extra control might generate missing data when merging datasets, if $j = 0$ is the specification with the extra control, $\theta_{0|d_0 = 1}$ would be the effect of $D_0$ on $y_0$ in the sub-population for which all covariates of specification $0$ are observable. If $j \neq 0$, the interpretation is analogous, but instead of a causal effect, $\theta_{j | d_j  = 1}$ is simply the linear projection coefficient. 

The conclusion is that, when comparing robustness checks, we are comparing coefficients \textit{conditional on their respective sub-populations}. Any extrapolation from that requires explicit assumptions on missing data, which the researcher may or may not be willing to make. In light of this distinction, the previous null should be rewritten as:

\begin{align}
	\label{eq:null_redefined}
	H_0 : & \max_{j \in \{1, \ldots, m\}}|\theta_{0|d_0 = 1} - \theta_{j | d_j = 1}|  \leq b  \notag  \\
	&  \iff 
	 \begin{pmatrix}
		|\theta_{0|d_0 = 1} - \theta_{1 | d_1 = 1}| \\
		\vdots \\
		|\theta_{0|d_0 = 1} - \theta_{m|d_m = 1}|
	\end{pmatrix} \leq  \begin{pmatrix}
	b \\ \vdots \\ b
	\end{pmatrix},
\end{align}
where the inequality is component-wise.

The null hypothesis has a form resembling typical nulls from the moment inequalities literature, except that these population conditions are not yet written as moment conditions. 

I test this by using the procedures of \cite{coxandshi}, which is set up for moment inequalities. Their proposed test builds on \cite{mohamadetal}, and uses a quasi-likelihood statistic, which it compares to a chi-square distribution with degrees of freedom corresponding to the number of moment inequalities violated in sample; crucially for this application, and differently from \cite{mohamadetal}, the test allows for an intercept. Notably, \cite{coxandshi}'s test has uniformly asymptotically exact size and no tuning parameters. 

To apply the test, we rewrite $H_0$ as:
\begin{align}
	\label{eq:redef_A}
	A \begin{pmatrix}
		\theta_{0|d_0 = 1} \\
		\vdots \\
		\theta_{m|d_m = 1}
	\end{pmatrix} \leq \begin{pmatrix}
	b \\ \vdots \\ b
	\end{pmatrix},
\end{align}
where $A$ is a matrix with rows $i = 1, \ldots, 2m$, columns $j = 1, \ldots, m+1$ and elements $A_{ij} = (-1)^{i-1}$ for $j = 1$, $A_{(2j - 3)j} = -1$ and $A_{(2j - 2)j} = 1$ for $j> 1$, $A_{ij} = 0$ otherwise. Matrix $A$ simply reproduces the absolute values inequalities in Equation \ref{eq:null_redefined}. The right-hand side vector of stacked $b$'s is augmented such that it has appropriate dimensions.

Consider the Frisch–Waugh–Lovell (FWL) decomposition to express $\theta_{j|d_j = 1} $ as:
\begin{align}
	\label{eq:fwl}
	\theta_{j|d_j = 1}  = \E \left[\frac{u_j y_j}{\E[u_j^2 | d_j = 1]}\,|d_j = 1\right], 
\end{align}
where $u_j$ is the residual of projecting $D_j$ onto $X_{M_j}$.
We would ideally like to define moment functions 
\begin{equation}
	\label{eq:def_m}
	\mathfrak{m}_j(u_{ji}, y_{ji}, d_{ji}) = \frac{u_{ji} y_{ji}}{\E[u_{ji}^2 | d_{ji} = 1]},
\end{equation}
but those are only meaningfully defined in the subset where $d_{ij} = 1$, and cannot be stacked in a single unconditional expectation due to differences in the conditioning events. 

However, when considering linear regressions such as in Equation \ref{eq:regs}, we can write
\begin{align}
	\label{eq:cond_uncond_mof}
	\theta_{j|d_j = 1}  = \E \left[\frac{u_j y_j}{\E[u_j^2 | d_j = 1]}\,|d_j = 1\right] = \E \left[\frac{\tilde{u}_j \tilde{y}_j}{\E[\tilde{u}_j^2 ]}\right], 
\end{align}
where $\tilde{u}_j \equiv u_j d_j$, $\tilde{y}_j \equiv y_j d_j$. The second equality is relevant and deserves elaboration: it rewrites a conditional expectation as an unconditional one, as a device to allow for missing data when applying the full vector version of the conditional chi-square test (CC henceforth) or its refined version (RCC) from \cite{coxandshi}. While they also provide a version for subvector inference in conditional moment inequality models (and an example of how to use it in linear regressions), it is not directly applicable here since each specification has a different informational set to condition on. 

From Equation \ref{eq:cond_uncond_mof}, we can define a function $\tildem_j(u_{ji}, y_{ji}, d_{ji})$ such that $\E[\mathfrak{m}_j(u_{ji}, y_{ji}, d_{ji}) | d_{ij} = 1] = \E[\tildem_j(u_{ji}, y_{ji}, d_{ji})]$:
\begin{equation}
	\label{eq:def_m_tilde}
	\tildem_j (u_{ji}, y_{ji}, d_{ji}) = \frac{\tilde{u}_{ji} \tilde{y}_{ji}}{\E[\tilde{u}_{ji}^2 ]}.
\end{equation}  
Since FWL is a result first and foremost about the algebra of least squares, regression estimators can be rewritten as 
\begin{equation}
	\label{eq:def_htheta}
	\htheta_j = \E_n \left[\frac{\tilde{u}_j \tilde{y}_j}{\E_n[\tilde{u}_j^2 ]}\right],
\end{equation}
 where $\E_n[.]$ denotes sample mean. Notice I write the moment functions with known denominators (as a proof device), but that is not the case for the actual estimator. Let $\htheta$ denote the stacked estimators $\htheta_j$'s. 

In the following, let $|.|$ denote matrix determinant and $\epsilon, K$ be fixed positive quantities. Values of $\epsilon$ can vary across assumptions.
Let $\tilde{\sigma}_j^2 := \operatorname{var}\left(\tildem_j (u_{ji}, y_{ji}, d_{ji})\right)$. 

\begin{assumption} \label{assum:1}

\begin{enumerate}[a)]
	\item \label{assum:1a}  $\left\{(y_i, X_i, d_i)\right\}$ are i.i.d. across $i = 1, \ldots, n$, where $X_i =
	\begin{pmatrix}
		(D_{0i}X_{0i})'      &              &           \\
		& \ddots       &           \\
		&              & (D_{mi}X_{mi})'
	\end{pmatrix}$, $d_i = (d_{0i} \ldots d_{mi})'$ , $y_i = (y_{0i} \ldots y_{mi})'$.
	\item \label{assum:1b} $\tilde{\sigma}_j^2 > 0$, $\forall j =0, \ldots, m$,
	\item \label{assum:1c} $\big| \operatorname{corr}\bigbraces{\tildem (u_{i}, y_{i}, d_{i})}\big| > \epsilon $, where $\operatorname{corr}\bigbraces{\tildem (u_{i}, y_{i}, d_{i})}$ is the correlation matrix across all $j = 0, \ldots, m$ moment conditions,
	\item \label{assum:1d} $\E\bigsquarebrackets{\big|\tildem_j (u_{ji}, y_{ji}, d_{ji})/\tilde{\sigma}_j \big|^{2 + \epsilon} }\leq K$ for $j = 0, \ldots, m$, 
	\item \label{assum:1e} $n_j/n \overset{p}{\to} p_j > \epsilon$, where $n_j = \sum_{i = 1}^n d_{ji}$,
	\item \label{assum:1f} $\E[\tildeu_j^2]^2 > \epsilon$, $\forall j =0, \ldots, m$.
\end{enumerate}
\end{assumption} 

 Assumptions \ref{assum:1}\ref{assum:1a},\ref{assum:1b} are standard, assumption \ref{assum:1}\ref{assum:1c} is a mild regularity condition which is typical in the moment inequality literature, and \ref{assum:1}\ref{assum:1d} is a standard $2+\epsilon$ moment requirement. Assumption \ref{assum:1}\ref{assum:1e} states that the probability of observing all covariates from specification $j$ is bounded away from $0$, which seems reasonable in most applications and is similar in spirit to an overlap assumption for propensity score estimators. Assumption \ref{assum:1}\ref{assum:1f} is also a regularity condition needed to apply the test beyond moment functions in Equation \ref{eq:def_m_tilde}, and to the estimators defined in Equation \ref{eq:def_htheta}. The next proposition ensures that, under the stated assumptions, the CC and RCC tests ($\phi_n^{CC}(b, \alpha)$ and $\phi_n^{RCC}(b, \alpha)$, respectively) are valid, by showing
Theorem 1 of \cite{coxandshi} applies:

\begin{prop}
	\label{prop:RCC_test}
	Suppose Assumption \ref{assum:1} holds. Let $F \in \mathcal{F}$ denote the data distribution giving rise to $\htheta$ and $b(\theta) = \max_{j \in \{1, \ldots, m\}} |\theta_0 - \theta_j|$. 
	Then, Theorems 1, 4 and 5 of \cite{coxandshi} hold, and:
	\begin{enumerate}[a)]
		\item $\limsup _{n \rightarrow \infty} \sup _{F \in \mathcal{F}} \mathbb{E}_{F} \phi_{n}^{\mathrm{CC}}(b(\theta), \alpha) \leq \alpha$.
		\item $\limsup _{n \rightarrow \infty} \sup _{F \in \mathcal{F}} \mathbb{E}_{F} \phi_{n}^{\mathrm{RCC}}(b(\theta), \alpha) \leq \alpha$
		\item For matrix $A$ defined in Equation \ref{eq:redef_A} and for a sequence $\{(F_n): F_n \in \mathcal{F} \}$, if for all $j \in \{1, \ldots, m\}$ $\sqrt{n}\left(|\htheta_{jn} - \htheta_{0n} -b(\theta_n) \right) \to 0$,  $\lim_{n \rightarrow\infty} \mathbb{E}_{F_n} \phi_{n}^{\mathrm{RCC}}(b(\theta_n), \alpha) = \alpha$.
	\end{enumerate}
\end{prop}
\begin{proof}
	See Appendix \ref{sec:prop_main_proof}.
\end{proof}

Parts a) and b) show the CC and RCC tests are asymptotically uniformly valid, and part c) shows when all inequalities bind (or are sufficiently close to binding), $\phi_n^{RCC}(b(\theta), \alpha)$ does not under-reject asymptotically. The original paper\footnote{Differently from \cite{coxandshi}, I omit the sequence and the supremum over the identified set, since this setting is point identified.} also has a formal statement of results for the ``irrelevance of distant inequalities"  (IDI) property: If some inequalities are very slack, the test effectively ignores them. In particular, if all but one inequality is very slack, the test becomes the one-sided t-test for the sole binding inequality, which is uniformly most powerful. For the purpose of this work, it suffices to know the test works for all OLS-like estimators.

The test statistic\footnote{\cite{coxandshi} write $T_n$ as a function of $\theta$ instead of $b$, because $b$ is taken as given and $\theta$ is the varying parameter for which a confidence interval is constructed by test inversion; whereas here, I am effectively proposing the construction of a one-sided confidence interval for $b$.} for the CC test is then:
\begin{align}
	\label{eq:test_statistic}
	T_n(b) = \min_{\mu: A\mu \leq b} n
	(\htheta - \mu)' (\hat{\Sigma})^{-1} (\htheta - \mu),
\end{align}
where $\hat{\Sigma}$ is a suitable\footnote{In the sense of Assumption \ref{assum:2} , in the Appendix.} covariance matrix estimator.

Define $\hat{r}(b)$ to be the number of active inequalities in sample, that is,
\begin{align}
	\label{eq:r_hat_def}
	\hat{r}(b)= \sum_{s = 1}^{2m} \mathbf{1} \{a_s \hat{\mu} = b\},
\end{align}
with $a_s$ being the $s-th$ row of $A$.
$T_n(b)$ asymptotically follows a chi-square distribution with $\hat{r}(b)$ degrees of freedom. Hence, the test rejects at significance level $\alpha$ 
if $T_n(b) > \chi^2_{\hat{r}(b), (1-\alpha)}$. The robustness radius is then defined as:
\begin{align}
	\label{eq:rob_radius_def}
	b_{RR}(\alpha) = \min\{b \geq 0: T_n(b) \leq \chi^2_{\hat{r}(b), (1-\alpha)}\}.
\end{align}
The RCC test only differs from the CC test when $\hat{r}(b) = 1$. While I implement the RCC test in examples, I will skip the theoretical details, which can be found in \cite{coxandshi}'s Supplementary Appendix A.1.

The test can be easily adapted to incorporate equality between the main estimate and robustness checks, if the researcher thinks that would be appropriate (see \cite{lu2014robustness} for details of when this should hold). This is done by setting the appropriate entries of vector $b$ to $0$. In fact, a very interesting aspect of the \cite{coxandshi} test is that, when the whole vector $b$ is equal to $0$, it reduces to the robustness test shown in \cite{lu2014robustness}. 

Additionally, we can define intuitive concepts:
\begin{defn}[Full Robustness]
	We say $\htheta_0$ is fully robust if $b_{RR}(\alpha) = 0$.
\end{defn}
Full robustness, as previously discussed, is a strong requirement. A natural definition to determine overall robustness of $\htheta_0$ is to verify if $b_{RR}$ is not enough to overturn the sign of the main specification:
\begin{defn}[Robust with respect to sign]
	We say $\htheta_0$ is robust with respect to sign if $b_{RR}(\alpha) < |\htheta_0|$.
\end{defn}
Depending on applications, the researcher might think robustness with respect to sign is too loose of a requirement and may want to compare $b_{RR}$ with other context-specific quantities. See the application in Section \ref{sec:application} for an example. Nevertheless, robustness with respect to sign gives a baseline to form conclusions.

	\section{Simulations and Illustrations}
\label{sec:simulations}
I use simulations to illustrate some simple properties of $b_{RR}$. For ease of exposition, I assume no missing data, such that $\htheta_j \xrightarrow{p} \theta_j$ for all $j$. I also take the covariance matrices to be known.

\begin{rem}
	\label{rem:Q1}
	The robustness radius is not necessarily $0$ even if robustness checks estimates fall inside the confidence interval centered at the main estimate.
	This can be seen in a simple example: ignore true probability limits $\theta_0, \, \theta_1$. Suppose only two estimates, $\htheta_0 = 0$, $\htheta_1 = 1.5$, $\sigma_0 = \sigma_1 = 1$. A $95\%$ CI for $\theta_0$ covers $\htheta_1$. However, in this simple example, testing with $b = 0$ is equivalent to a simple test of difference in means: $H_0: \big|\theta_0 - \theta_1\big| = 0$. Whether the CI for the difference covers $0$ depends on the standard error of the difference, $\sqrt{2 - 2\rho_{01}}$, where $\rho_{01}$ is the correlation between estimators. We can see the role $\rho_{01}$ plays in Table \ref{tbl:Q1}.
\end{rem}

\begin{table}[ht]
\centering
\begin{tabular}{rrrrrr}
  \hline
$\rho_{01}$ & $0$ & $0.5$ & $0.8$ & $0.9$ & $0.99$ \\ 
  \hline
 & 0.000 & 0.000 & 0.360 & 0.754 & 1.267 \\ 
   \hline
\end{tabular}
\caption{Illustrative example: $b_{RR}$ for different correlation values
$\rho_{01}$, with $\hat{\theta}_0 = 0$, $\hat{\theta}_1 = 1.5$,
$\sigma_0 = \sigma_1 = 1$, $\alpha = 0.05$.} 
\label{tbl:Q1}
\end{table}

\begin{rem}
	\label{rem:Q2}
	The robustness radius is not necessarily $0$ with probability $1-\alpha$ even if the distance between true probability limits is less than $c_{\alpha} \sigma_0$, where $c_{\alpha}$ denotes the Normal critical value, and even if estimators are uncorrelated.
\end{rem}
Now, we think about the true probability limits $\theta_0$, $\theta_1$ giving rise to the data. Suppose $\theta_0 = 0$, $\theta_1 = 1.5$, $\sigma_0 = \sigma_1 = 1$, $\rho_{01} = 0$. A $95\%$ confidence interval centered at the truth $\theta_0$ covers $\theta_1$. Again, testing with $b = 0$ is equivalent to a simple test of difference in means: $H_0: \big|\theta_0 - \theta_1\big| = 0$ against $H_1: \big|\theta_0 - \theta_1\big| > 0$. But the null is not true, so the probability of non-rejection is the probability of a Type II error, which is: $\prob\Bigbraces{|\htheta_0 - \htheta_1|/\sqrt{2} \leq 1.96}  = \Phi\Bigbraces{ 1.96 - 1.5/\sqrt{2}} = 0.815$. By running the simulations described in the left panel of Figure \ref{fig:b_RR_hist_rho0}, I get $b_{RR} = 0$ with probability $0.808$. The smaller number is due to the fact that with only one active inequality (which, by construction, is the case when $m = 1$), the refinement of the RCC test makes it more powerful than a two-sided test for a difference in means (essentially reducing to a one-sided test in this case)\footnote{For details on the RCC test, see Supplemental Appendix A.1 of \cite{coxandshi}.}.
Furthermore, the left panel of Figure \ref{fig:b_RR_hist_rho0} shows the distribution of the robustness radius across simulations.

\begin{rem}
	The robustness radius adapts to the variance of the robustness checks.
\end{rem}
Again, this is easy to see from the example of two estimators. Consider $\theta_0 = 0$, $\theta_1 = 1.5$, $\sigma_0 = 1, \sigma_1 = \sqrt{2}, \rho_{01} = 0$. Now, the variance of $\htheta_0 - \htheta_1$ is $3$. The probability of non-rejection is:   $\prob\Bigbraces{|\htheta_0 - \htheta_1|/\sqrt{3} \leq 1.96}  = \Phi\Bigbraces{1.96 - 1.5/\sqrt{3}} = 0.863$. Simulations described in the right panel of Figure \ref{fig:b_RR_hist_rho0} give $b_{RR} = 0$ with probability $0.864$. Although visual inspection of Figure \ref{fig:b_RR_hist_rho0} does not suggest distributions of robustness radius change much upon doubling the variance of $\htheta_1$, Figure \ref{fig:b_RR_hist_rho0.9} tells a different story: when $\rho_{01}$ increases to $0.9$, the distribution of $b_{RR}$ becomes much more right-skewed upon increasing $\sigma_1$. This should align with researchers' expectations, that often get non-significant results for robustness checks, but don't believe this to be enough of a signal of non-robustness: more sampling uncertainty on the robustness checks (higher $\sigma_1$) results on smaller $b_{RR}$ on average, and this effect is more salient if $\htheta_0, \htheta_1$ are highly correlated. 

\begin{figure}[h!]
	\centering
	
	\begin{subfigure}{0.48\textwidth}
		\centering
		\includegraphics[width=\linewidth]{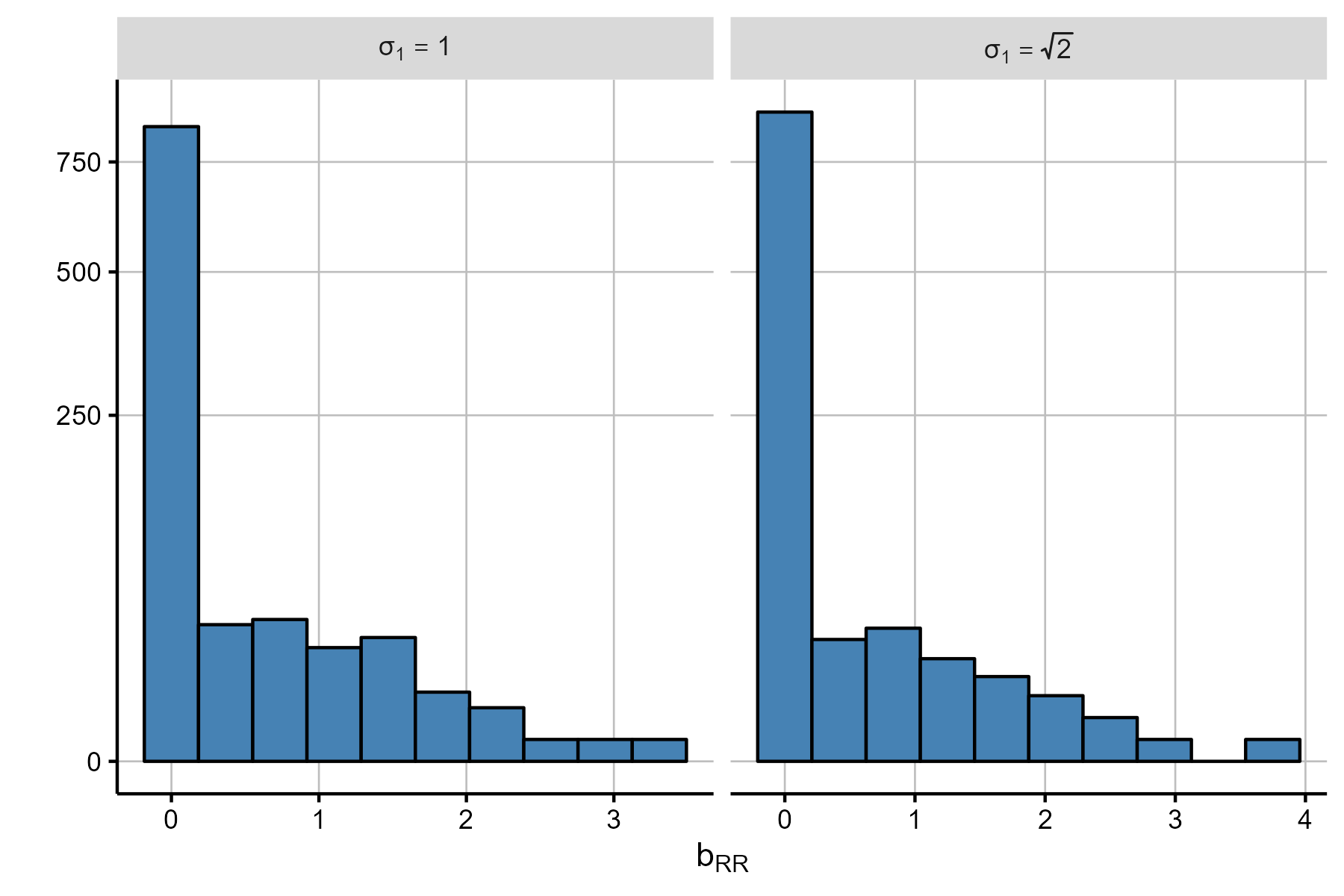}
		\caption{$\rho_{01} = 0$}
		\label{fig:b_RR_hist_rho0}
	\end{subfigure}
	\hfill
	\begin{subfigure}{0.48\textwidth}
		\centering
		\includegraphics[width=\linewidth]{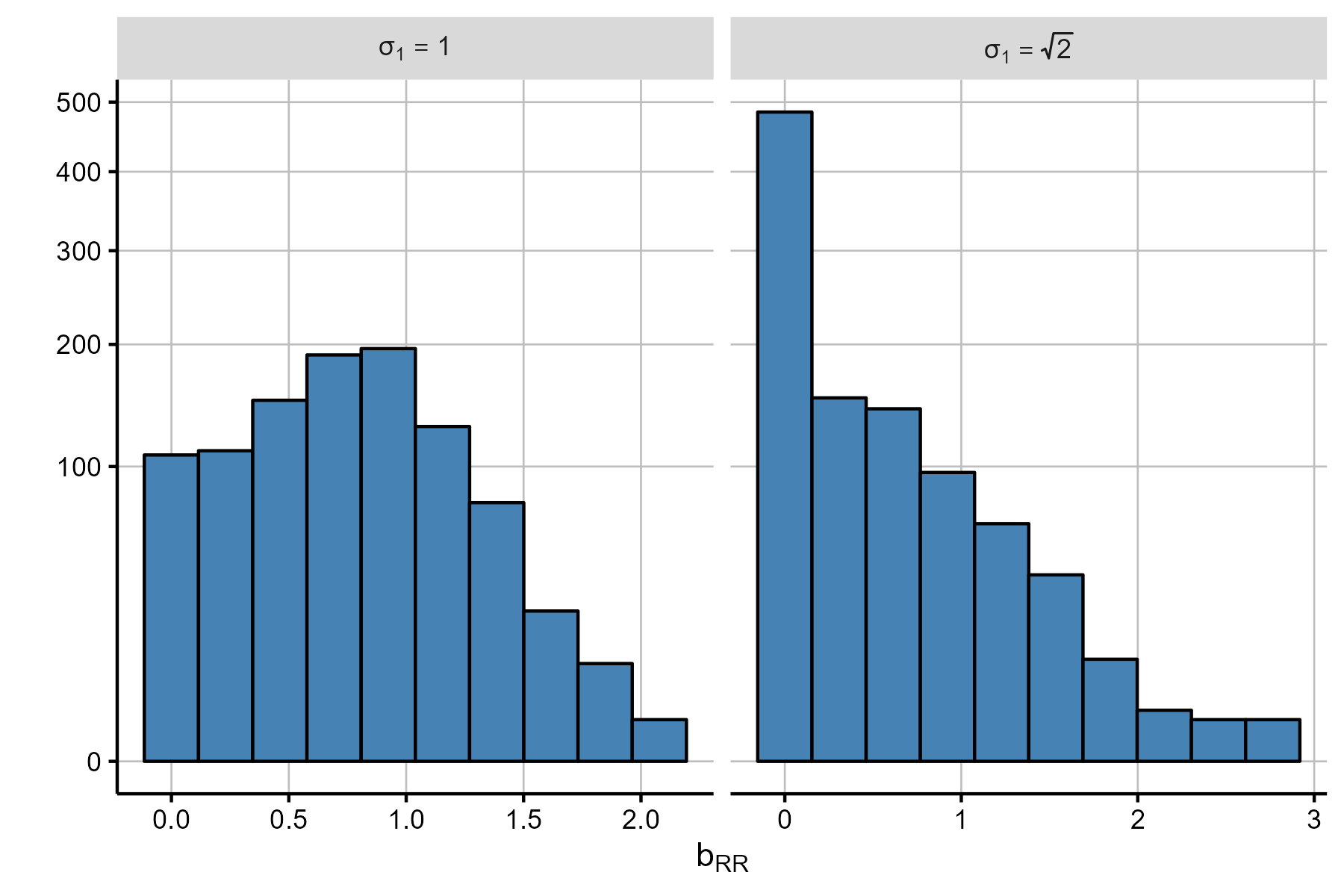}
		\caption{$\rho_{01} = 0.9$}
		\label{fig:b_RR_hist_rho0.9}
	\end{subfigure}
	
	\caption{Histogram of the robustness radius with different variance matrices. $\theta_0 = 0$, $\theta_1 = 1.5$, $\sigma_0 = 1$, $1{,}000$ simulations, $\alpha = 0.05$.}
\end{figure}

So far, with only two estimands, the robustness radius could be computed with a simple t-test. In most applications, there likely are more estimands, and so a test like \cite{coxandshi} is required. I perform two exercises to elucidate the behavior of the robustness radius with $m > 1$. Before explaining them, it should be noted that the framework considers $m$ to be fixed, because this is the condition under which the test is provably valid. The researcher should know which robustness checks to run \textit{before} computing the robustness radius. For a reference in moment inequalities testing where $m$ is allowed to grow, see \cite{10.1093/restud/rdy065}.

The first exercise computes the average $b_{RR}$ across $1,000$ simulations using specifications from Section 5.1 of \cite{coxandshi}. $\Omega_N$, $\Omega_Z$, $\Omega_P$ denote negative, neutral and positive correlation structures, respectively. Three data generating processes (DGPs) are considered: normal, student-t, and mixed normal distributions. The maximum distance between true parameters, $\max_j|\theta_0 - \theta_j|$, is set to $2.5$, and $n = 100$. For each number of robustness checks $m$ and distribution type, $m$ true parameter structures are considered; for example, for $m = 3$ we have $(\theta_0, \theta_1, \theta_2, \theta_3) \in \left\{(0, 0, 0, 2.5), (0, 0, 2.5, 2.5), (0, 2.5, 2.5, 2.5)\right\}$; analogously for other $m$ values. As noted in Table \ref{tbl:Q1}, the positive correlation structure increases precision of the robustness radius. 

\begin{figure}[ht]
	\centering
	\includegraphics[width=0.8\textwidth]{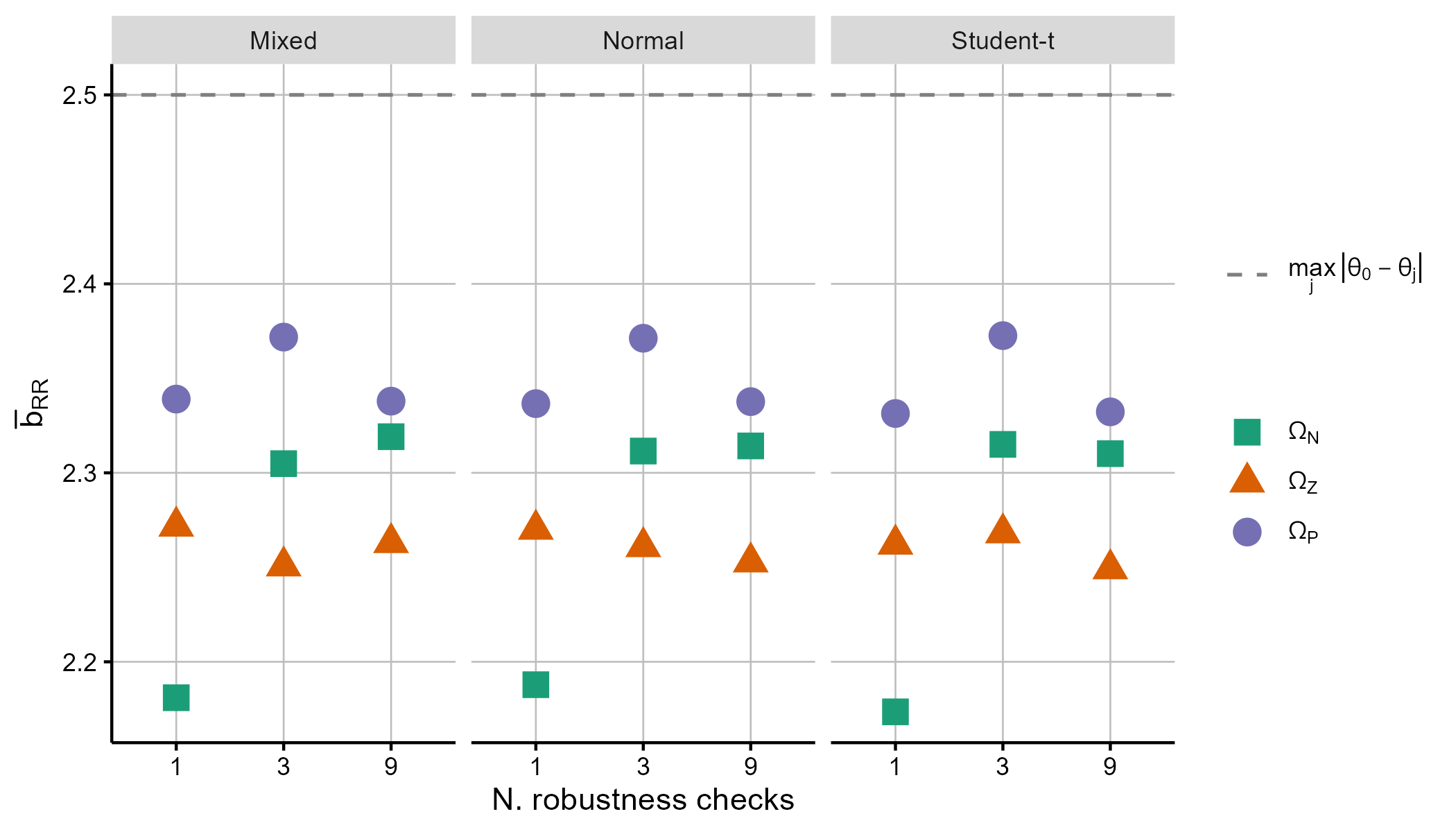}
	\caption{Average $b_{RR}$ following specifications in \cite{coxandshi}.}
	\label{fig:bRR_by_num_rob_checks_cox_and_shi_specs}
\end{figure}

However, in these examples, the estimands are well-separated, which is often not the case in robustness analyses. 
With that in mind, and with the purpose of expanding on Properties \ref{rem:Q1} and \ref{rem:Q2} for the cases of $m>1$, I perform a limit experiment to study the behavior of $b_{RR}$ when $\max_j|\theta_0 - \theta_j| = 1.5 < c_{\alpha} \sigma_0$, where I set $\sigma_0 = 1$. Estimators $\htheta_0, \ldots, \htheta_m$ are taken to be exactly normal with $1= \sigma_0 = \ldots = \sigma_m$, and I focus on positive correlation structures (for simplicity, I assume the same correlations between all estimators). 

Figure \ref{fig:bRR_by_num_rob_checks_robustness_specs} shows the results. Correlation structure plays a much more relevant role in precision of $b_{RR}$ than the well-separated case. Furthermore, $b_{RR}$ remains fairly stable across different numbers of robustness checks (I maintain the $m$ true parameter structures considered as in Figure \ref{fig:bRR_by_num_rob_checks_cox_and_shi_specs}). This exercise shows that, when $\max_j|\theta_0 - \theta_j|$ is small with respect to sampling uncertainty $c_{\alpha} \sigma_0$, higher correlations imply values of $b_{RR}$ which are closer to $\max_j|\theta_0 - \theta_j|$; this result remains true regardless of $m$. 

\begin{figure}[ht]
	\centering
	\includegraphics[width=0.8\textwidth]{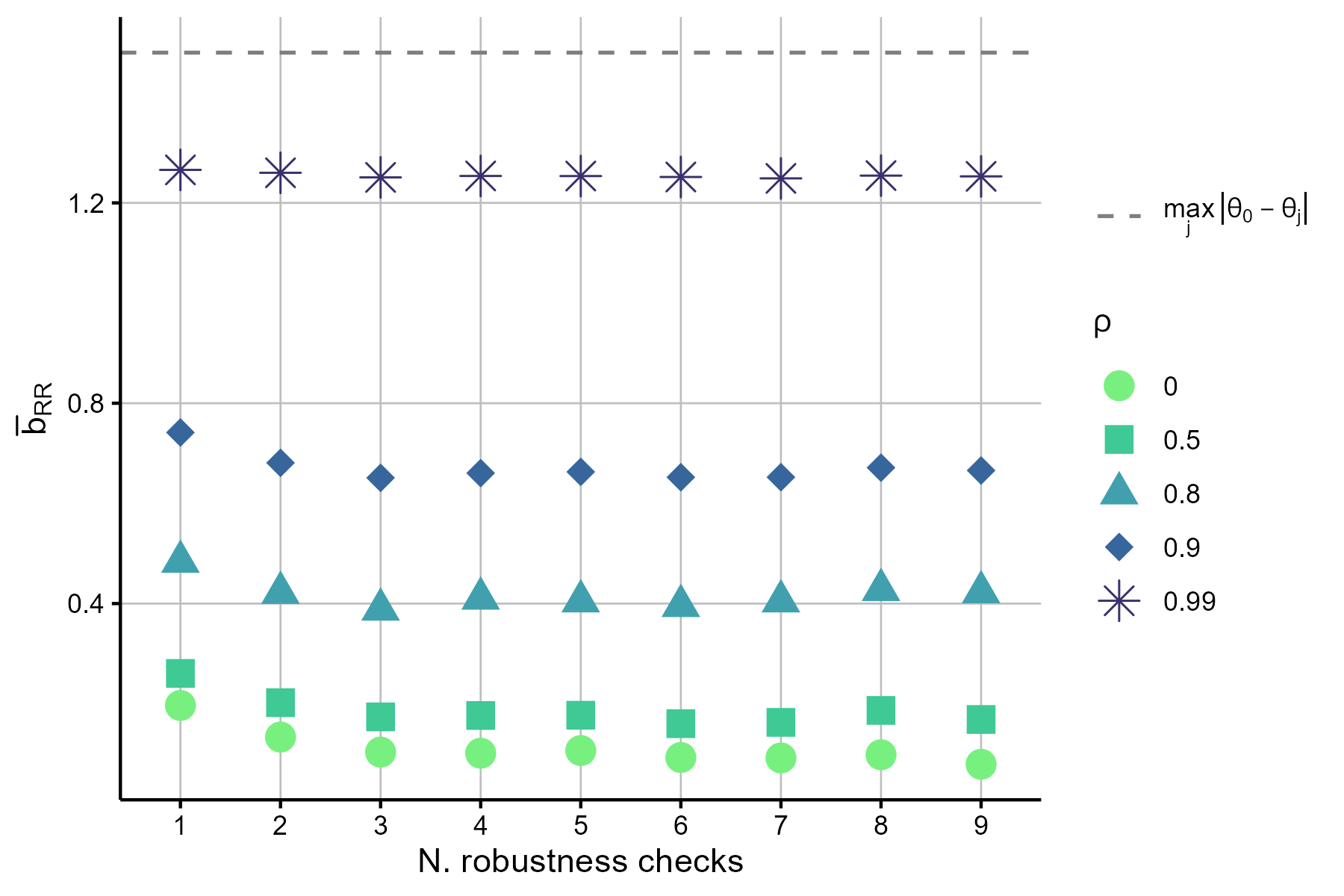}
	\caption{Average $b_{RR}$ for different correlations $\rho$, and different $m$. $\sigma_j = 1$ for all $j$.}
	\label{fig:bRR_by_num_rob_checks_robustness_specs}
\end{figure}

	\section{Application: ``Wage Stagnation and the Decline of Standardized Pay Rates, 1974–1991"}
\label{sec:application}

\cite{massenkoffwilmers} examine how transitioning from wage-setting practices based on job titles and seniority to managerial discretion affects wages; they find wage declines, particularly among lowest-paid workers. In a context of declining worker bargaining power, they conclude the transition to discretionary pay amplified blue-collar workers' real wage losses during the 1980s. While they do not claim standardized pay rates to be the primary driver of wage stagnation, the argument is that the mechanism had acted as a barrier against market pressures during the shift from high to low worker bargaining power. The authors use microdata from the Wage Fixing Authority Survey (WFAS), which contains information on pay levels and pay-setting practices for 50,000 workplaces surveyed between 1974 and 1991. 

I replicate\footnote{Results vary slightly, because I use R instead of Stata; in R, I use package \texttt{fixest} to estimate the high-dimensional fixed effects specifications, using function \texttt{feols} with argument \texttt{fixef.rm = "singleton"}, to drop singleton fixed effects. The algorithm is such that this is done in one pass. The authors use Stata's \texttt{reghdfe}, which instead removes singletons recursively. This is the cause of the difference in the number of observations in columns (4) and (5). But estimators do not have any visible differences. See \url{https://github.com/lrberge/fixest/issues/478} for more details on the algorithm. I stick to R's \texttt{fixest} because it is significantly faster than Stata's \texttt{reghdfe}, which is relevant when bootstrapping to get covariances between robustness checks estimates.} Table 2 of the paper, which is the main table exploring the effects of non-standardized pay on wages. Results can be seen in Table \ref{tbl:replication_massenkorf_wilmers_main}.

\begin{table}[htbp]
   \bigskip
   \centering
   \begin{tabular}{lccccc}
      \toprule
                                                  & (1)           & (2)           & (3)           & (4)           & (5)\\  
      \midrule 
      Nonstandardized Pay                         & -0.145        & -0.108        & -0.077        & -0.008        & -0.010\\   
                                                  & (0.003)       & (0.003)       & (0.002)       & (0.001)       & (0.004)\\   
      log(Workers at Pay Level)                   & 0.044         & 0.040         & 0.034         & 0.013         & 0.012\\   
                                                  & (0.002)       & (0.001)       & (0.001)       & (0.000)       & (0.001)\\   
      log(Workers in Establishment)               & 0.050         & 0.048         & 0.040         & 0.014         & 0.012\\   
                                                  & (0.001)       & (0.001)       & (0.001)       & (0.002)       & (0.004)\\   
      log(Workers in Job)                         & -0.016        & -0.023        & -0.012        & -0.016        & -0.013\\   
                                                  & (0.001)       & (0.001)       & (0.001)       & (0.001)       & (0.002)\\   
      log(Minimum Wage)                           & 0.085         & 0.075         & 0.001         & 0.010         & -0.001\\   
                                                  & (0.011)       & (0.010)       & (0.010)       & (0.004)       & (0.005)\\   
      Collective Bargaining                       &               & 0.034         & 0.021         & 0.002         & -0.003\\   
                                                  &               & (0.004)       & (0.003)       & (0.002)       & (0.003)\\   
      Share Managerial, Clerical in Establishment &               & 0.018         & -0.008        & -0.002        & -0.018\\   
                                                  &               & (0.006)       & (0.005)       & (0.004)       & (0.011)\\   
      Coworkers’ Occupational Level               &               & 0.556         & 0.262         & 0.006         & 0.002\\   
                                                  &               & (0.008)       & (0.006)       & (0.004)       & (0.009)\\   
      Union Density in Industry-Wage Area         &               & 0.049         & 0.030         & -0.002        & 0.008\\   
                                                  &               & (0.005)       & (0.004)       & (0.002)       & (0.006)\\   
       \\
      Observations                                & 900,359       & 852,024       & 829,626       & 772,809       & 548,889\\  
      R$^2$                                       & 0.12977       & 0.21547       & 0.72946       & 0.96546       & 0.96128\\  
      Within R$^2$                                & 0.12239       & 0.20878       & 0.19237       & 0.01064       & 0.00885\\  
       \\
      Year fixed effects                          & $\checkmark$  & $\checkmark$  &               &               & \\  
      Occupation-Year-City-Ind. fixed effects     &               &               & $\checkmark$  & $\checkmark$  & $\checkmark$\\   
      Occupation-Establishment fixed effects      &               &               &               & $\checkmark$  & $\checkmark$\\   
      Firm-Occupation-Year fixed effects          &               &               &               &               & $\checkmark$\\   
      \bottomrule
   \end{tabular}
   \caption{\label{tbl:replication_massenkorf_wilmers_main} Wage Effects of Nonstandardized Pay Rates}
\end{table}

Additionally, two tables on explicit robustness checks are also replicated for specifications (3) and (4) in Table \ref{tbl:replication_massenkorf_wilmers_main}. Table \ref{tbl:replication_massenkorf_wilmers_weights} replicates\footnote{Similarly to the previous footnote, number of observations differ, but there is no visible difference in the results. } Table A.2 of their paper, which checks for robustness across different weights. In Table \ref{tbl:replication_massenkorf_wilmers_main}, each job (establishment-by-occupation) is treated as one equal unit. In Table \ref{tbl:replication_massenkorf_wilmers_weights}, observations are weighted by different survey weights and number of employers. The authors argue their choice of not using survey weights in the main specifications are due to lack of documentation about how they are constructed. About robustness, they conclude: \textit{``The results are largely consistent with the job-weight results presented in the main tables. The exception is the employee-weighted job-by-establishment fixed effect model. This is driven entirely by the large jobs (above 31 workers) that make up 5 percent of the observations but 54 percent of the employment-weighted sample"}.

 The qualitative nature of the claim highlights the potential usefulness of the robustness radius measure. Columns 1-4 of Table \ref{tbl:massenkoff_wilmers_bRR}\footnote{Results for $b_{RR}$ used a covariance matrix computed through a hierarchical bootstrap as described in \cite{cameronmiller_guide} Section B, since specifications are  clustered at the establishment level, but some of them include fixed effects of the interaction between establishment and occupation. The number of replications is $1,000$. I use the trimmed bootstrap covariance estimator from Equation \ref{eq:trimmed_boot} to handle clustering across different sub-samples more straightforwardly, but other estimators may be available. Trimming is necessary as the proof of Proposition \ref{prop:RCC_test} relies on the ratio representation of the OLS estimator. See \cite{hansen2022econometrics} Section 10.14 for details, and for a warning about the widespread use of the untrimmed bootstrap.} show the robustness radius for different combinations of the robustness checks present in Table \ref{tbl:replication_massenkorf_wilmers_weights}. Because the authors are not explicit about whether the preferred specification is the one with just occupation-year-city-industry fixed effects, or also with occupation-establishment fixed effects (columns (3) and (4) of Table \ref{tbl:replication_massenkorf_wilmers_main}, respectively), and  in fact run robustness checks for both specifications, I also compute $b_{RR}$ separately for each of them. Indeed, by looking at columns (2) and (4) of Table \ref{tbl:massenkoff_wilmers_bRR}, excluding the largest $5\%$ of jobs when running the employee-weighted job-by-establishment fixed effect model reduces the robustness radius almost in half. Still, it remains about $38\%$\footnote{Formally defining a measure of robustness relative to its main specification faces challenges that I refrain from in this paper, but it is still informative to look at these ratios heuristically.} of the main coefficient (see columns (4) of Tables \ref{tbl:replication_massenkorf_wilmers_main} and \ref{tbl:massenkoff_wilmers_bRR}) --which does not seem to be a threat to the claim of robustness, but is much higher when compared to the ratio of $6.7\%$ of model (3).

\begin{table}[htbp]
   \bigskip
   \centering
   \begin{adjustbox}{width = 1.2\textwidth, center}
      \begin{tabular}{lcccccccccc}
         \toprule
          & \multicolumn{2}{c}{Inv. Row Wt} & \multicolumn{2}{c}{Survey Wt.} & \multicolumn{2}{c}{CPS Wt.} & \multicolumn{4}{c}{N. Employees} \\ \cmidrule(lr){2-3} \cmidrule(lr){4-5} \cmidrule(lr){6-7} \cmidrule(lr){8-11}
                                                 & (1)           & (2)           & (3)           & (4)           & (5)           & (6)           & (7)           & (8)           & (9)           & (10)\\  
         \midrule 
         Nonstandardized Pay                     & -0.077        & -0.008        & -0.071        & -0.011        & -0.076        & -0.008        & -0.070        & -0.002        & -0.073        & -0.006\\   
                                                 & (0.002)       & (0.001)       & (0.003)       & (0.002)       & (0.003)       & (0.002)       & (0.003)       & (0.003)       & (0.002)       & (0.002)\\   
         Controls                                & $\checkmark$  & $\checkmark$  & $\checkmark$  & $\checkmark$  & $\checkmark$  & $\checkmark$  & $\checkmark$  & $\checkmark$  & $\checkmark$  & $\checkmark$\\   
         Excl. largest $5\%$ of job brackets     &               &               &               &               &               &               &               &               & $\checkmark$  & $\checkmark$\\   
          \\
         Observations                            & 829,626       & 772,809       & 829,626       & 772,809       & 779,649       & 726,170       & 829,626       & 772,809       & 785,896       & 731,742\\  
         R$^2$                                   & 0.72946       & 0.96546       & 0.71875       & 0.93325       & 0.75053       & 0.96622       & 0.83179       & 0.97303       & 0.77752       & 0.96441\\  
         Within R$^2$                            & 0.19237       & 0.01064       & 0.13896       & 0.00364       & 0.19487       & 0.00897       & 0.20574       & 0.01784       & 0.17861       & 0.00831\\  
          \\
         Occupation-Year-City-Ind. fixed effects & $\checkmark$  & $\checkmark$  & $\checkmark$  & $\checkmark$  & $\checkmark$  & $\checkmark$  & $\checkmark$  & $\checkmark$  & $\checkmark$  & $\checkmark$\\   
         Occupation-Establishment fixed effects  &               & $\checkmark$  &               & $\checkmark$  &               & $\checkmark$  &               & $\checkmark$  &               & $\checkmark$\\   
         \bottomrule
      \end{tabular}
   \end{adjustbox}
   \caption{\label{tbl:replication_massenkorf_wilmers_weights} Standardized pay rates effects robustness to alternative weights}
\end{table}

Table \ref{tbl:replication_massenkorf_wilmers_merit} replicates\footnote{Also no visible difference.} Table A.3 and checks for robustness across different definitions of flexible pay-settings. In their words, \textit{``In the main results we define flexible pay as pay variation due to pure merit; a combination of merit and seniority; or other, nonseniority reasons for pay variation. All of these survey responses indicate managerial discretion and potentially individualized pay. However, other reasonable approaches would be to exclude other, nonseniority reasons for pay variation or exclude combination systems from the definition of merit pay. Online Appendix Table A.3 shows that this coding decision has little impact on the wage effect estimates: flexible or nonstandardized pay effects are consistently negative across variable definitions."} Once again, the robustness radius measure shows itself useful in being able to determine what ``little" means: 
for the two specifications (3) and (4) discussed before, both robustness radiuses are about $15.8\%$ of their main specifications. Importantly, all robustness radiuses are not $0$: following the test suggested in \cite{lu2014robustness} would lead to rejection of robustness, but one would hardly believe different measures of merit (or weight) identify the same effect. One would rather think they identify estimands that are ``close,'' and that closeness can be made precise.

\begin{table}[htbp]
   \bigskip
   \centering
   \begin{tabular}{lcccccc}
      \toprule
                                              & (1)           & (2)           & (3)           & (4)           & (5)           & (6)\\  
      \midrule 
      Merit, Combination, or Other            & -0.077        & -0.008        &               &               &               &   \\   
                                              & (0.002)       & (0.001)       &               &               &               &   \\   
      Merit or Combination                    &               &               & -0.075        & -0.009        &               &   \\   
                                              &               &               & (0.002)       & (0.001)       &               &   \\   
      Merit (Narrow)                          &               &               &               &               & -0.064        & -0.008\\   
                                              &               &               &               &               & (0.002)       & (0.001)\\   
      Controls                                & $\checkmark$  & $\checkmark$  & $\checkmark$  & $\checkmark$  & $\checkmark$  & $\checkmark$\\   
       \\
      Observations                            & 829,626       & 772,809       & 829,626       & 772,809       & 829,626       & 772,809\\  
      R$^2$                                   & 0.72946       & 0.96546       & 0.72851       & 0.96546       & 0.72482       & 0.96546\\  
      Within R$^2$                            & 0.19237       & 0.01064       & 0.18951       & 0.01076       & 0.17851       & 0.01051\\  
       \\
      Occupation-Year-City-Ind. fixed effects & $\checkmark$  & $\checkmark$  & $\checkmark$  & $\checkmark$  & $\checkmark$  & $\checkmark$\\   
      Occupation-Establishment fixed effects  &               & $\checkmark$  &               & $\checkmark$  &               & $\checkmark$\\   
      \bottomrule
   \end{tabular}
   \caption{\label{tbl:replication_massenkorf_wilmers_merit} Effects robustness to alternative operationalizations of flexible pay-setting}
\end{table}

The last two columns of Table \ref{tbl:massenkoff_wilmers_bRR} should be enough to judge overall robustness; at the very least, sign is preserved. However, while full robustness is generally a strong requirement, sign preservation seems to be a loose one when checking against different weighting schemes or definitions of flexible pay-settings; these can be considered minor choices and a sign flip would be concerning. 

Another practical recommendation is to compare the robustness radius computed across explicit robustness checks to the radius computed considering specifications from Table \ref{tbl:replication_massenkorf_wilmers_main}. It is common for researchers to include in the main table a few columns which clearly lack important controls and, as such, are explicitly biased. The robustness radius which considers only seemingly small perturbations (like different weighting schemes and definitions of  flexible pay-settings) should at least be smaller than the radius which considers explicitly biased regressions. That is indeed the case: including regressions presented in Table \ref{tbl:replication_massenkorf_wilmers_main}, the radiuses are $0.069$ and $0.137$, for columns (3) and (4) respectively, very far from the true robustness radiuses presented in Table \ref{tbl:massenkoff_wilmers_bRR}.

The application highlights that context matters---not only when choosing which robustness checks to run, but also when judging the magnitude of the robustness radius. Nonetheless, it also highlights how a standardized measure facilitates forming such judgments.  

 Ideally, there will be a single preferred specification to compare robustness checks against. If the researcher is not willing to make that statement, perhaps they should consider a framework for model uncertainty and use methods such as in \cite{xinyuebei2024}, as discussed in Section \ref{sec:introduction}. Finally, it is worth noticing that the second row of Table \ref{tbl:massenkoff_wilmers_bRR}, giving the maximum distance between estimates, is remarkably close to the robustness radius for every column. This is not surprising in this application, since each coefficient had originally very small standard errors, and it speaks to good power properties of the RCC test. As previously seen in simulations, the robustness radius may be significantly smaller than the maximum distance between estimates, depending on the standard errors and correlation structure.

\begin{table}[ht]
\centering
\begin{tabular}{rrrrrrrrr}
  \hline
 & Wt. (3) & Wt. (4) & Wt. (3) -5\% & Wt. (4) -5\% & Merit (3) & Merit (4) & All (3) & All (4) \\ 
  \hline
$b_{RR}$ & 0.00661 & 0.00596 & 0.00515 & 0.00303 & 0.01207 & 0.00119 & 0.01207 & 0.00596 \\ 
  $\max_j|\htheta_0 - \htheta_j|$ & 0.00663 & 0.00597 & 0.00517 & 0.00304 & 0.01209 & 0.00119 & 0.01209 & 0.00597 \\ 
   \hline
\end{tabular}
\caption{Robustness radius and maximum distance of estimates
for the regressions from tables A.2 and A.3 of
\cite{massenkoffwilmers}, centered at columns (3) and (4) of Table (
\ref{tbl:replication_massenkorf_wilmers_main}). ``All'' refers to specifications 
in both tables A.2 and A.3, together.  -5\% refers to excluding large jobs.} 
\label{tbl:massenkoff_wilmers_bRR}
\end{table}

	\section{Connection with Sensitivity Measures}

\label{subsec:sensitivity}

The robustness radius can be explicitly connected to sensitivity measures of omitted variable bias (OVB) such as in \cite{altonjicatholic}, \cite{oster2019unobservable}, \cite{cinelli2020making}, \cite{chernozhukov2022long}, \cite{diegert2023assessingomittedvariablebias}.  Despite substantial advances in the aforementioned papers regarding interpretability of sensitivity parameters, it can still be challenging to reason about their values.
While the literature on formal sensitivity analysis cautions that there is no inherent motive to assume that bias from observables is sufficient to inform about bias from unobservables, if such an assumption is justified, linking the robustness radius to sensitivity parameters can assist researchers in choosing appropriate values of the latter. I perform this exercise for the sensitivity measure of \cite{diegert2023assessingomittedvariablebias}, which is scalar-valued.

Consider the baseline model $j = 0$. I follow their notation (except I use $\tau$ instead of $r$ to avoid confusion with degrees of freedom defined in Equation \ref{eq:r_hat_def}) and rename coefficients with the superscript ``med'' and ``long'' (for medium and long regressions, respectively). $\theta_0^{\text{med}} = \theta_0$ in the sense that it is the main specification, but now the parameter of interest is $\theta_0^{\text{long}}$, which is only partially identified. 

Let $X_0 = \{X_k: k \in M_0\}$ for simplicity:
\begin{equation}
	y_0 = \alpha_0^{\text{med}} + \theta_0^{\text{med}} D_0 +  (\beta_{0}^{\text{med}})' X_0 + e_0^{\text{med}}.
\end{equation}
In contrast, let $W$ be a vector of unobserved variables of dimension $p_W$, such that the researcher is interested in $\theta_0^{\text{long}}$ in the regression
\begin{equation}
	y_0 = \alpha_0^{\text{long}} + \theta_0^{\text{long}} D_0 +  (\beta_{0}^{\text{long}})' X_0 + \gamma'W + e_0^{\text{long}}.
\end{equation}
Next, consider the regressions of $D_0$ on $X_0$, and of $D_0$ on $X_0$ and $W$:
\begin{equation}
	\label{eq:D_med}
	D_0 = \alpha_{D_0} ^{\text{med}}+  (\xi^{\text{med}})' X_0+ e_{D_0}^{\text{med}},
\end{equation}
\begin{equation}
	D_0 = \alpha_{D_0} ^{\text{long}}+  (\xi^{\text{long}})' X_0+ \pi' W + e_{D_0}^{\text{long}},
\end{equation}
where $ \alpha_{D_0} ^{\text{med}}, \, \alpha_{D_0} ^{\text{long}}$ are the constant terms and $ e_{D_0}^{\text{med}}, \, e_{D_0}^{\text{long}}$ are the orthogonal residuals.

\cite{diegert2023assessingomittedvariablebias} refer to $\sqrt{\text{var}(\pi' W)}$ as ``selection on unobservables" and $\sqrt{\text{var}\left((\xi^{\text{long}})'X_0\right)}$ as ``selection on observables"; since currently ``selection on observables" is used in reference to robustness checks, I instead call $\sqrt{\text{var}\left((\xi^{\text{long}})'X_0\right)}$ ``selection on included covariates".

As per Theorem 4 of \cite{diegert2023assessingomittedvariablebias}, if $\pi \neq 0$, there is selection on unobservables and $ \theta_0^{\text{long}} \in  \left[\theta_0^{\text{med}} - b^{\text{UN}}, \theta_0^{\text{med}} + b^{\text{UN}} \right]$, where $b^{\text{UN}}$ is the bias from unobservables, given by
\begin{equation}
	\label{eq:b_diegert}
	b^{\text{UN}} = \begin{cases}
		\sqrt{\frac{\operatorname{var}\left(e_0^{\text{med}}\right)}{\operatorname{var}\left(e_{D_0}^{\text{med}}\right)}
			\cdot
			\frac{\bar{\tau}_D^2 R^2_{D \sim X_0}}{1 - R^2_{D \sim X_0} - \bar{r}_D^2}},
		\text{ if } 0 \leq \bar{\tau}_D^2 \leq \sqrt{1 - R^2_{D \sim X_0}},\\
		\infty, \text{ oth.},
	\end{cases}
\end{equation}
with $R^2_{D \sim X_0}$ being the $R^2$ of the regression in Equation \ref{eq:D_med}, and $ \bar{\tau}_D^2$ is the researcher-specified parameter that limits selection on unobservables relative to selection on included covariates: $\frac{\sqrt{\operatorname{var}\left(\pi'W\right)}}{\sqrt{\operatorname{var}\left((\xi^{\text{long}})' X_0\right)}} \leq \bar{\tau}_D$. 

Now, I ask: what is the value of $\bar{\tau}_D$ that would make $b^{\text{UN}} = b_{RR}$, and most importantly, what is the interpretation of such a backward engineering exercise? Simple algebra gives the first answer:
\begin{equation}
	\widehat{\bar{\tau}}_D = \sqrt{\frac{\left(b_{RR}^2 \cdot \frac{\operatorname{var}\left(e_{D_0}^{\text{med}}\right)}{\operatorname{var}\left(e_0^{\text{med}}\right)}\right) \left(1 - R^2_{D \sim X_0}\right)}{R^2_{D \sim X_0} + b_{RR}^2 \cdot \frac{\operatorname{var}\left(e_{D_0}^{\text{med}}\right)}{\operatorname{var}\left(e_0^{\text{med}}\right)}}},
\end{equation}
While quantities $\frac{\operatorname{var}\left(e_{D_0}^{\text{med}}\right)}{\operatorname{var}\left(e_0^{\text{med}}\right)}$, $R^2_{D \sim X_0}$ are not known, they can be estimated by their sample counterparts.

As for the second answer, I have previously stated in Section \ref{sec:introduction} that if the researcher is worried about uncertainty in $\theta_0$ and wishes to use robustness checks to expand the model, confidence intervals such as in \cite{xinyuebei2024} should be used instead. Estimating $b^{\text{UN}}$ by $b_{RR}$ does not yield a valid confidence interval for a partially identified parameter. However, $\widehat{\bar{\tau}}_D$ does have a clear interpretation. To see it, denote $b^{\text{OB}} := \max_{j \in \{1, \ldots, m\}} |\theta_0 - \theta_j|$ the bias from observables. Suppose we would like to test $H_0: b^{\text{OB}} \leq b^{\text{UN}}$; or analogously, $H_0: \theta_j \in [\theta_0 -b^{\text{UN}}, \theta_0 + b^{\text{UN}}] \,\,\,\forall j \in \{1, \ldots, m\}$ --- in words, that all robustness checks estimands fall within $b^{\text{UN}}$ of the main estimand $\theta_0$. 

With that in mind, notice by definition \ref{def:rob_rad} that 
\begin{align*}
	b_{RR} = \min\{b^{\text{UN}}: \text{ fail to reject }H_0: b^{\text{OB}} \leq b^{\text{UN}}\},
\end{align*}
that is, the smallest value of $b^{\text{UN}}$ for which we cannot reject bias from observables is less than or equal to bias from unobservables is precisely the robustness radius $b_{RR}$. Hence, $\widehat{\bar{\tau}}_D$ is the smallest value of selection on unobservables relative to selection on included covariates that one must assume to not reject the null that bias from observables is at most as large as bias from unobservables. If the researcher deems $\widehat{\bar{\tau}}_D$ to be large (small) relative to a pre-specified $\bar{\tau}_D$, it means the null would (not) be rejected, with the conclusion that robustness checks are (not) informative of selection on unobservables. Alternatively, if the null is believed to be true, which is reasonable when robustness checks are thought as being small perturbations around the main specification, the researcher should be picking $\bar{\tau}_D$ at least as large as $\widehat{\bar{\tau}}_D$.

Although this exercise is performed for the sensitivity measure of \cite{diegert2023assessingomittedvariablebias}, 
it can potentially be used for other scalar measures of sensitivity $\tau$ as well. The precise interpretation of $\hat{\tau}$ will be specific to each sensitivity parameter's definition.

	\section{Conclusion}
\label{sec:conclusion}
This paper introduces the robustness radius as a formalized measure for evaluating the stability of robustness checks in empirical research. The proposed metric provides a structured and interpretable framework for quantifying the extent to which alternative specifications deviate from the main estimand, adapting to the variance and covariance structure of the estimators. I show that a testing procedure recently proposed by \cite{coxandshi} can be used to ensure that this measure overcomes the limitations of ad-hoc comparisons and the often informal interpretation of robustness checks as ``qualitatively similar."

While testing the strict null hypotheses of identical estimands might lead to rejection, the robustness radius offers a more adaptive perspective, accounting for inherent variability in the checks. The application to wage-setting practices in \cite{massenkoffwilmers} showcases how the robustness radius can translate variability across different specifications, and how we can measure robustness instead of rejecting it. Interpretation of robustness is still context-specific; nonetheless, a standardized measure makes analyses more transparent.
 
The robustness radius is not a measure of omitted variable bias, but it can be linked to sensitivity parameters under such settings. This connection can help researchers fine tune their perception about both the choice of checks and sensitivity parameters (in a process that should, of course, be correctly documented).

	\newpage

\appendix

\section{Proof of Proposition \ref{prop:RCC_test}}
\label{sec:prop_main_proof}

I show that Assumption \ref{assum:1} directly implies Assumption 3 of \cite{coxandshi} for unconditional moment function $\tildem$, which in turn implies Assumption 3 of \cite{coxandshi} holds for estimator vector $\htheta$, a sufficient condition for Theorem 5 of  \cite{coxandshi} to hold and hence Proposition \ref{prop:RCC_test} to hold. The main reason for this proof is to show the test is applicable even with estimands that slightly differ for each $j$ (namely, by having expectations conditional on $d_j$ = 1). Let $F \in \mathcal{F}$ denote a possible distribution of $\htheta$, and let $\Theta_0(F)$ denote the identified set for $\theta$ under $F$ (in the current setting, $\Theta_0(F)$  is just a point).

  I omit the dependence of $\htheta$ on $\theta$ for ease of notation.
\begin{assumption}[Cox and Shi Assumption 3]
	\label{assum:2}
	The given sequence $\left\{\left(F_{n}, \theta_{n}\right): F_{n} \in \mathcal{F}, \theta_{n} \in \Theta_{0}\left(F_{n}\right)\right\}_{n=1}^{\infty}$ satisfies, for every subsequence, $n_{m}$, there exists a further subsequence, $n_{q}$, and there exists a sequence of positive definite $d_{m} \times d_{m}$ matrices, $\left\{D_{q}\right\}$ such that, for moment function $\m^{*}$:
	
	\begin{enumerate}[a)]
		\item Under the sequence $\left\{F_{n_{q}}\right\}_{q=1}^{\infty}$,
		
		\begin{equation}\label{eq:unif_normality}
			\sqrt{n_{q}} D_{q}^{-1 / 2}\left(\E_{n_q}{[\m}_{n_{q}}^{*}]-\E_{F_{n_q}}[\m^{*}]\right) \rightarrow_{d} N(\mathbf{0}, \Omega) 
		\end{equation}
		
		for some positive definite correlation matrix, $\Omega$, and
		
		\begin{equation}\label{eq:consistent_cov_est}
			\left\|D_{q}^{-1 / 2} \widehat{\Sigma}_{n_{q}} D_{q}^{-1 / 2}-\Omega\right\| \rightarrow_{p} 0 . 
		\end{equation}
		
		\item $\Lambda_{q} A\left(\theta_{n_{q}}\right) D_{q} \rightarrow \bar{A}_{0}$ for some $d_{A} \times d_{m}$ matrix $\bar{A}_{0}$, and for every $J \subseteq\left\{1, \ldots, d_{A}\right\}$, $\operatorname{rk}\left(I_{J} A\left(\theta_{n_{q}}\right) D_{q}\right)=\operatorname{rk}\left(I_{J} \bar{A}_{0}\right)$, where $\Lambda_{q}$ is the diagonal $d_{A} \times d_{A}$ matrix whose $j$-th diagonal entry is one if $e_{j}^{\prime} A\left(\theta_{n_{q}}\right)=\mathbf{0}$ and $\left\|e_{j}^{\prime} A\left(\theta_{n_{q}}\right) D_{q}\right\|^{-1}$ otherwise.
	\end{enumerate}
\end{assumption}

\begin{proof}
	I omit the arguments of the moment functions to ease notation.
Firstly, we assume that $\E[\tilde{u}^2]$ is known. By 
the arguments of the proof of
\cite{coxandshi} Theorem 4, Assumption \ref{assum:1} directly implies Assumption \ref{assum:2} holds for $\m^{*} = \tildem$ as defined in Equation \ref{eq:def_m_tilde}, with
	 $\Sigma_n = \operatorname{var}_{F_n}\bigbraces{\tildem (u_{i}, y_{i}, d_{i})}$, $D_n = \operatorname{diag}\left(\Sigma_n\right)$,  $\Omega_n = \operatorname{corr}_{F_n}\bigbraces{\tildem (u_{i}, y_{i}, d_{i})}$,
	  $A$ as in Equation \ref{eq:redef_A}, and $\hat{\Sigma}_n $ being some covariance estimator (which I omit, to later prove the trimmed bootstrap covariance estimator defined in Equation \ref{eq:trimmed_boot} satisfies Equation \ref{eq:consistent_cov_est}). 
	
	From Equation \ref{eq:cond_uncond_mof},  we can rewrite Equation \ref{eq:unif_normality} as:
	\begin{equation}\label{eq:unif_normality_theta}
		\sqrt{n_{q}} D_{q}^{-1 / 2}\left(\E_{n_q}{[\tildem}_{n_{q}}]-\theta_{|d=1(n_q)}\right) \rightarrow_{d} N(\mathbf{0}, \Omega) \text{ or } 	\sqrt{n_{q}} \left(\E_{n_q}{[\tildem}_{n_{q}}]-\theta_{|d=1(n_q)}\right) \rightarrow_{d} N(\mathbf{0}, \Sigma).
	\end{equation}
	
	Now suppose $\E[\tildeu^2]$ is unknown, but can be estimated by its sample analogue. We can define the vector function
	\begin{equation}
		g:\mathbb{R}^{2(m+1)}\to\mathbb{R}^{(m+1)}, \,\,	g(v_1,\ldots,v_m,\,s_1,\ldots,s_m)=
		\begin{pmatrix}
			g_0(v_0,s_0)\\[1mm]
			\vdots\\[1mm]
			g_m(v_m,s_m)
		\end{pmatrix},\quad\text{with}\quad
		g_j(v_j,s_j)= v_j\,\frac{\E[u_j^2 d_j]}{s_j}.
	\end{equation}
	By evaluating the function at $v _j= v_{n,j} = \E_n[\tildem_j] \text{ and } s_j = s_{n,j} = \E_n[u_j^2 d_j]$, we get:
	\begin{equation}
	     g_j(v_{n,j},s_{n,j}) = \htheta_j,
	\end{equation}
	and by evaluating it at true probability limits $v_j = \E[\tildem_j]$, $s_j = \E[u_j^2d_j]$, we get:
	\begin{equation}
		g_j(\E[\tildem_j],\E[u_j^2d_j]) =\theta_{j|d_j=1(n_q)}.
	\end{equation}
	We then apply the delta method. Firstly, notice that the joint asymptotic distribution is given by 
	\begin{equation}
		\sqrt{n_q}\left(
		\begin{pmatrix}v_{n_q,0}\\  s_{n_q,0} \\ \vdots\\ v_{n_q,m}\\ s_{n_q,m}\end{pmatrix}
		-\begin{pmatrix}\theta_{0|d_0=1(n_q)} \\ \E[u_0^2 d_0] \\ \vdots\\ \theta_{m|d_m=1(n_q)}\\ \E[u_m^2 d_m]\end{pmatrix}
		\right)
		\to_d N(0,\Psi),
	\end{equation}
	where the variances elements are well defined due to assumptions \ref{assum:1}\ref{assum:1b},\ref{assum:1f}.
	
	 Because the function $g$ is separable, the Jacobian is block diagonal with each block $j$ being the gradient (row) vector:
	\begin{equation}
			\nabla g_j =\begin{pmatrix}
				\frac{\partial g_j}{\partial v_j}\Big|_{(\theta_{j|d_j=1(n_q)},\E[u_j^2 d_j])} \\ 
				\frac{\partial g_j}{\partial s_j}\Big|_{(\theta_{j|d_j=1(n_q)},\E[u_j^2 d_j])}
			\end{pmatrix} = 
		    \begin{pmatrix} 
			 	1 \\ 
			 	-\dfrac{\E[\mathfrak{m}_j |d_j = 1]}{\E[u_j^2 d_j]} 
			\end{pmatrix},
	\end{equation}
where $\E[u_j^2 d_j]$ is bounded away from zero by assumptions \ref{assum:1} \ref{assum:1f}. By the multivariate delta method,
\begin{equation}
	\label{eq:final_norm}
	\sqrt{n_q}\left(\htheta-\theta_{|d=1(n_q)}\right) \to_d  N\left(0,\; \operatorname{diag}(\{\nabla g_j\}_{j = 0}^m)\,\Psi\,\operatorname{diag}(\{\nabla g_j\}_{j = 0}^m)'\right).
\end{equation}
This proves Equation \ref{eq:unif_normality} Assumption \ref{assum:2} a) holds for our estimator $\htheta$. 

Now, to show Equation \ref{eq:consistent_cov_est} holds, consider the trimmed bootstrap covariance estimator, that is, for $B$ bootstrap samples of $\htheta_n$, $\htheta_n^{\text{boot}}$, define

\begin{equation}
	\htheta_n^{\text{trimboot}} = \htheta_n^{\text{boot}}\mathbf{1}\{||\htheta_n^{\text{boot}}|| \leq \tau_n\},
\end{equation}

\begin{equation}
	\label{eq:trimmed_boot}
	\widehat{\Sigma}_n^{\text{trimboot}}= \frac{1}{B-1}\sum_{b = 1}^B(\htheta_n^{\text{trimboot}} - \E_b[\htheta_n^{\text{trimboot}}])(\htheta_n^{\text{trimboot}} - \E_b[\htheta_n^{\text{trimboot}}])',
\end{equation}
where $\E_b[\htheta_n^{\text{trimboot}}]$ is the bootstrap sample average of $\htheta_n^{\text{trimboot}}$, and $\tau_n \to \infty$ is a sequence of positive trimming numbers satisfying $\tau_n^4 = O\left(e^{n^{1/2}}\right)$. Then, the definition of the OLS estimator together with Assumption \ref{assum:1} imply Theorem 10.12 of \cite{hansen2022econometrics} applies, that is, $\htheta_n^{\text{trimboot}}$ is consistent. 

Finally, notice that Assumption \ref{assum:2} b) is trivially satisfied as in this application of the test, matrix A does not depend on $\theta$. This concludes the proof.

One important remark is that estimators for the standard errors of individual specifications --that is, the usual number reported inside (.) in tables-- relies on the asymptotic result:
\begin{equation}
	\sqrt{n_j}(\htheta_j - \theta_{j|d_j=1} | d_j = 1) \to_{d} N\left(0, \operatorname{Avar (\htheta_j |d_j = 1)}\right).
\end{equation}
The reported standard error is conceptually the same as the one derived from Equation \ref{eq:final_norm}:
\begin{equation}
	\sqrt{\frac{\widehat{\operatorname{Avar}}(\htheta_j |d_j = 1)}{n_j}} = \sqrt{\frac{\widehat{\operatorname{Avar}}(\htheta_j |d_j = 1)}{\frac{n_j}{n}}\frac{1}{n}} = \sqrt{\frac{\widehat{\operatorname{Avar}}(\htheta_j )}{n}} ,
\end{equation} 
where the result follows from assuming $\frac{n_j}{n} \overset{p}{\to} p_j$, the fact that $\htheta_j = \htheta_j | d_j = 1$ (from definition in Equation \ref{eq:def_htheta}) and Slutsky's theorem. 

\end{proof}

\newpage
\bibliographystyle{ecta}
\bibliography{tex/references.bib}
\end{document}